\documentclass[10pt, ruled, linesnumbered, vlined]{article}%

\usepackage{amsfonts}
\usepackage{color}
\usepackage[usenames, dvipsnames]{xcolor}
\usepackage{fullpage}
\usepackage{subfigure}
\usepackage{amsmath}
\usepackage{amsthm}
\usepackage{comment}
\usepackage[dvipdfmx]{graphicx}
\usepackage{epstopdf} 
\usepackage{amssymb, epsfig}
\usepackage[authoryear]{natbib}
\usepackage{amssymb}
\usepackage{authblk}
\usepackage[font=footnotesize, labelfont=bf]{caption}
\usepackage{algorithm2e}

\usepackage{multirow}

\newtheorem{thm}{Theorem}
\newtheorem{prop}{Proposition}
\newtheorem{asmp}{Assumption}

\newtheorem{defn}{Definition}
\newtheorem{lem}{Lemma}

\newtheorem{cor}{Corollary}

\makeatother

\DeclareMathOperator*{\esssup}{ess\,sup}

\begin{document}

\title{On the Use of Penalty MCMC for Differential Privacy}
\author{Sinan Y{\i}ld{\i}r{\i}m \\
Faculty of Engineering and Natural Sciences \\
Sabanc{\i} University, \.{I}stanbul, Turkey}
\date{\today}
\maketitle

\begin{abstract}
We view the penalty algorithm of \citet{Ceperley_and_Dewing_1999}, a Markov chain Monte Carlo (MCMC) algorithm for Bayesian inference, in the context of data privacy. Specifically, we study differential privacy of the penalty algorithm and advocate its use for data privacy. We show that in the simple model of independent observations the algorithm has desirable convergence and privacy properties that scale with data size. Two special cases are also investigated and privacy preserving schemes are proposed for those cases: (i) Data are distributed among several data owners who are interested in the inference of a common parameter while preserving their data privacy. (ii) The data likelihood belongs to an exponential family.
\end{abstract}

\section{Introduction} \label{sec: Introduction}

One concern about sharing datasets for use of third parties is violation of privacy of individuals who contribute to those datasets with their private information. Significant amount of research is devoted to address that concern by either analysing privacy preserving properties of already existing methods or developing novel methods that are useful for statistical inference and still preserve data privacy to a quantifiable extent at the same time.

One popular formulation of data privacy is \emph{differential privacy} \citep{Dwork_et_al_2006}. Differential privacy associates a random algorithm with a mapping from the space of data to a space of probability distributions for the output of the algorithm. It then quantifies the privacy of the algorithm roughly by some `difference' between two probability distributions corresponding to two `neighbour' data sets that differ by only one entry. This `difference' is expressed by two nonnegative parameters which are typically shown by $\epsilon$ and $\delta$, and the algorithm is said to be $(\epsilon, \delta)$-differentially private. Specifically, let $\mathcal{A}$ be a random algorithm that takes input from a set $\mathcal{X} \subseteq \mathbb{R}^{n}$ for some $n \geq 1$ and produces random outputs in $\mathcal{S}$. Let $h(X, Y)$ be the edit distance or \emph{Hamming distance} between data sets $X, Y \in \mathcal{X}$ i.e.\ number of entry-wise differences between $X$ and $Y$.
\begin{defn}
\emph{\textbf{Differential privacy:}} We call a randomised algorithm $\mathcal{A}$ with domain $\mathcal{X}$ and range $\mathcal{S}$ $(\epsilon, \delta)$-differentially private  if for any measurable set $S \subset \mathcal{S}$ and for all $X, Y \in \mathcal{X}$ such that $h(X, Y) \leq 1$, we have
\[
P(\mathcal{A}(X) \in S) \leq \exp(\epsilon) P(\mathcal{A}(Y) \in S) + \delta.
\]
where $\mathcal{A}(X)$ and $\mathcal{A}(Y)$ denotes random outputs of the algorithm using inputs $X$ and $Y$, respectively. If $\delta = 0$, then $\mathcal{A}$ is called $\epsilon$-differential private.
\end{defn}
Obviously, the smaller $\epsilon$ and $\delta$ are the better privacy is preserved, since smaller $\epsilon$ and $\delta$ correspond to a closer match between the probability distributions of $\mathcal{A}(X)$ and $\mathcal{A}(Y)$.

\subsection{Relevance of Monte Carlo methods} \label{sec: Relevance of Monte Carlo methods}

In Bayesian inference, the answer to an unknown parameter of interest is in the form of a probability distribution, namely the posterior distribution, denoted by $\pi$ throughout the paper. Monte Carlo methods, which constitute a great proportion of Bayesian inference methods, are used to generate random samples either exactly or approximately from $\pi$ when $\pi$ is intractable, i.e.\ it does not have tractable expectations of certain functions of the parameter of interest. Therefore, by its nature Monte Carlo methods produce random outputs and that makes them natural candidates for privacy preserving methods.

Arguably the most widely used Monte Carlo method for Bayesian inference is Markov chain Monte Carlo (MCMC) (see \citet{Robert_and_Casella_2004} for example), where one generates a sequence of samples from a Markov chain that has $\pi$ as its stationary distribution. Study of MCMC methods in relation to differential privacy has recently attracted researchers. For example, \citet{Wang_et_al_2015} presented a modified version of the stochastic gradient Langevin dynamics algorithm of \citet{Welling_and_Teh_2011} for differential privacy. Even more recently, \citet{Foulds_et_al_2016} discussed differential privacy of several MCMC methods. 

\subsection{Contribution} \label{sec: Contribution}
In this paper, we contribute to the literature on privacy preserving methods by studying differential privacy of a specific MCMC algorithm, namely the \emph{penalty algorithm} \citep{Ceperley_and_Dewing_1999}. The penalty algorithm is actually an approximation of the Metropolis-Hastings (MH) algorithm \citep{Hastings_1970}, one of the most popular MCMC methods. The approximation is due to the use of the log-acceptance ratio of the MH algorithm in normally distributed noise with known variance. 

A true implementation of the penalty algorithm requires being able to compute the log-acceptance ratio in the first place. That is why the penalty algorithm has so far been used with approximately normally distributed estimates of the log-acceptance ratio \emph{only} when this ratio is intractable. Indeed, without privacy concerns, it is always favourable to use the log-acceptance ratio rather than its noisy version. 

However, in a data privacy setting, the penalty algorithm becomes more relevant and potentially more useful than the MH algorithm, an observation that has surprisingly escaped attention so far. We advocate the use of the penalty algorithm for a few reasons: 
\begin{itemize}
\item The penalty algorithm is an \emph{exact} algorithm in that its stationary distribution remains $\pi$.
\item The penalty algorithm deliberately uses noisy functions of data, hence it has desirable privacy preserving properties. 
\item Under certain conditions, the differential privacy as well as convergence properties of the penalty algorithm scale well with data size.
\end{itemize}
The first reason is already known \citep{Ceperley_and_Dewing_1999, Nicholls_et_al_2012}; we will justify the last two in Sections \ref{sec: The penalty algorithm} and \ref{sec: Differential privacy of the penalty algorithm}. Specifically, we first show that the penalty algorithm, although inferior to the MH algorithm that it mimics, can preserve the uniform or geometric ergodicity properties (with slower rates) of the MH algorithm under some conditions. Then we proceed with the analysis of the differential privacy of the algorithm and show that the algorithm does have desired privacy properties that scale favourably with data size. We also investigate the use of the algorithm in a \emph{data sharing} scenario when data are distributed among different owners who want to infer the same $\theta$ conditional on all the available data .

Our analysis is similar to that in \citet{Wang_et_al_2015} in that our main theorem (Theorem \ref{thm: Differential privacy of penalty algorithm}) employs the results on the so called \emph{advanced composition} \citep{Dwork_et_al_2010} and  \emph{Gaussian mechanism} \citep{Dwork_and_Roth_2013}. One difference is that the algorithm discussed in \citet{Wang_et_al_2015} is only asymptotically exact but computationally scalable in data size due to subsampling; whereas our algorithm preserves the target density at all iterations but computationally not scalable in data size (unless the posterior distribution has a convenient form such as the exponential family, which we will discuss further in Section \ref{sec: Exponential families}). 

In a recent work by \citet{Foulds_et_al_2016}, differential privacy of certain MCMC methods for exponential family distributions are studied when they are used with the sufficient statistics of the data modified by Laplacian noise. The authors show that the proposed methods are asymptotically unbiased in terms of their target distribution. However, for any finite data size, the algorithms have bias. The penalty method, on the other hand, remains unbiased and it is applicable to more general distributions than exponential families.

\subsection{Outline} \label{sec: Outline}
The rest of the paper is organised as follows: In Section \ref{sec: The penalty algorithm}, we introduce the penalty algorithm for a general target distribution and we provide some statistical properties of the algorithm mainly in terms of convergence. The purpose of Section \ref{sec: The penalty algorithm} is mainly to assist the analysis in Section \ref{sec: Differential privacy of the penalty algorithm} where we work out the differential privacy of the algorithm, specifically we will explain what happens to $\epsilon$ and $\delta$ of the algorithm as the data size increases. We conclude in Section \ref{sec: Conclusion} with final discussions and a mention to future work.

\section{The penalty algorithm} \label{sec: The penalty algorithm}
Suppose we are given a target distribution $\pi$ for variable $\theta \in \Theta \subseteq \mathbb{R}^{d_{\theta}}$ for some $d_{\theta} \geq 1$ with probability density shown by $\pi(\theta)$ and a proposal kernel $q$ with density shown as $q(\cdot | \theta)$ for $\theta \in \Theta$. Under some mild conditions on $q$ with respect to $\pi$, the Metropolis-Hastings (MH) algorithm \citep{Hastings_1970} generates a reversible Markov chain $\{ \theta_{t} \}_{t \geq 0}$ with invariant density $\pi$ as follows: Start with $\theta_{0}$; for $t \geq 0$, given $\theta_{t} = \theta$, propose $\theta' \sim q(\cdot | \theta)$ and set $\theta_{t + 1} = \theta'$ (or `accept' $\theta'$) with probability
\begin{align*}
\alpha(\theta, \theta') = \min \left\{ 1, r(\theta, \theta') \right\},
\end{align*}
else set $\theta_{t + 1} = \theta'$ (or `reject'), where
\begin{align*}
r(\theta, \theta') = \frac{\pi(\theta') q(\theta | \theta')}{\pi(\theta) q(\theta' | \theta)},
\end{align*}
is called the acceptance ratio given the current and proposed values $\theta$ and $\theta'$.

The penalty algorithm of \citet{Ceperley_and_Dewing_1999} replaces the logarithm of the acceptance ratio $\lambda(\theta, \theta') := \log r(\theta, \theta')$ with an unbiased normal estimate $\hat{\lambda}_{\sigma}(\theta, \theta')$ with variance $\sigma^{2}(\theta, \theta') = \sigma^{2}(\theta', \theta)$ that takes its values from a symmetric function $\sigma^{2}: \Theta^{2} \rightarrow \left[ 0, \infty \right)$, i.e.\
\begin{equation*}
\hat{\lambda}_{\sigma}(\theta, \theta') \sim \mathcal{N} (\lambda(\theta, \theta'), \sigma^{2}(\theta, \theta'))
\end{equation*}
and accepts $\theta'$ with the modified acceptance probability
\[
\hat{\alpha}_{\sigma}(\theta, \theta') = \min \left\{ 1, \hat{r}_{\sigma}(\theta, \theta') \right\}
\]
where the modified acceptance ratio
\[
\hat{r}_{\sigma}(\theta, \theta') = \exp \left( \hat{\lambda}_{\sigma}(\theta, \theta') - \sigma^{2}(\theta, \theta')/2 \right)
\]
has the \emph{penalty term} $\sigma^{2}(\theta, \theta')/2$ as a compromise for using an estimate $\hat{\lambda}_{\sigma}(\theta, \theta')$ instead of $\lambda(\theta, \theta')$.

\subsection{Statistical properties} \label{sec: Statistical properties}

The penalty algorithm has the remarkable property that the Markov chain it generates still has $\pi$ as its invariant density, thus the algorithm still targets $\pi$; see \citet{Ceperley_and_Dewing_1999} for the original proof and \citet{Nicholls_et_al_2012} for an alternative proof that shows that the penalty algorithm is reversible with respect to $\pi$. 

\citet{Nicholls_et_al_2012} also showed (in a more general setting) that using $\hat{\alpha}_{\sigma}(\theta, \theta')$ is always inferior to using $\alpha(\theta, \theta')$ in the Peskun sense, i.e.\ for any $\theta, \theta'$ the expected acceptance probability 
\begin{equation*}
\alpha_{\sigma}(\theta, \theta') := \mathbb{E} \left[ \hat{\alpha}_{\sigma}(\theta, \theta') \right] \leq \alpha(\theta, \theta'),
\end{equation*} 
which can be verified via Jensen's inequality. In addition, we show here that in fact the algorithm's performance in the Peskun sense (i.e.\ $\alpha_{\sigma}(\theta, \theta')$) decreases monotonically in $\sigma^{2}(\theta, \theta')$. (A simple proof by working out the derivative of $\alpha_{\sigma}(\theta, \theta')$ is given in the Appendix.)
\begin{prop} \label{prop: performance w.r.t. penalty term}
For any $\theta, \theta' \in \Theta$, $\alpha_{\sigma}(\theta, \theta')$ is a decreasing function of $\sigma^{2}(\theta, \theta') \in [0, \infty)$.
\end{prop}

Denote the transition kernel of the MH algorithm and the corresponding penalty algorithm with variance $\sigma^{2}(\theta, \theta')$ by $T$ and $T_{\sigma}$, respectively. Despite its (perhaps not surprising) inferiorities that are mentioned above, $T_{\sigma}$ inherits the favourable ergodicity properties of the $T$ under some conditions. The following proposition on uniform ergodicity is proven in a more general setting in \citet{Nicholls_et_al_2012}; we present it here specifically for the penalty algorithm for completeness.
\begin{prop} \label{prop: ergodicity of the penalty algorithm}
Suppose that $\sup_{\theta, \theta'} \sigma^{2}(\theta, \theta') = B < \infty$. If the MH algorithm that targets $\pi$ with proposal $q$ is uniformly ergodic, then the penalty algorithm that uses $\sigma^{2}(\theta, \theta')$ for the variance of its noise and the same $q$ for the proposal is also uniformly ergodic.
\end{prop}
\begin{proof}
We need to show that the expected acceptance probability of the penalty algorithm satisfies
\begin{equation} \label{eq: minorisation condition for the penalty algorithm}
\alpha_{\sigma}(\theta, \theta') \geq \kappa \alpha(\theta, \theta')
\end{equation}
for some $\kappa > 0$ for all $\theta, \theta'$. Then we can follow the same steps starting from Equation (6.3) of \citet[Appendix A]{Nicholls_et_al_2012} to conclude that for any $\theta \in \Theta$ and measurable set $E \subseteq \Theta$ we have $T_{\sigma}(\theta, E) \geq \kappa T(\theta, E)$ and that since by hypothesis $T$ satisfies a minorisation condition so does $T_{\sigma}$, which leads to uniform ergodicity. To show \eqref{eq: minorisation condition for the penalty algorithm}, let $\kappa = 1 - \Phi(B/2)$, where $\Phi$ is the cumulative distribution function of $\mathcal{N}(0, 1)$, and observe that, for any $\theta, \theta' \in \Theta$, and $V \sim \mathcal{N}(0, 1)$, we have
\begin{align*}
P\left( \hat{\lambda}_{\sigma}(\theta, \theta') \geq \lambda(\theta, \theta') \right) &= P \left( \sigma(\theta, \theta') V - \sigma^{2}(\theta, \theta')/2 \right) \\
& = P \left( V \geq \sigma(\theta, \theta')/2 \right) \\
&\geq P \left( V \geq B/2 \right) \\
& = 1 - \Phi(B/2) = \kappa.
\end{align*}
Therefore, by using a decomposition for the expected value, we have
\begin{align*}
\alpha_{\sigma}(\theta, \theta') &= \mathbb{E}\left[ \hat{\alpha}(\theta, \theta') \right] \\
& \geq \mathbb{E}\left[ \hat{\alpha}_{\sigma}(\theta, \theta') \left\vert \hat{\lambda}_{\sigma}(\theta, \theta') \geq \lambda(\theta, \theta') \right. \right] P\left( \hat{\lambda}_{\sigma}(\theta, \theta') \geq \lambda(\theta, \theta') \right) \\
& \geq \alpha(\theta, \theta') \kappa
\end{align*}
where the second inequality holds since $\hat{\alpha}_{\sigma}(\theta, \theta') \geq \alpha(\theta, \theta')$ in the event $\hat{\lambda}_{\sigma}(\theta, \theta') \geq \lambda(\theta, \theta')$.
\end{proof}

Assumption of uniform ergodicity for the MH algorithm can be restrictive. Thus, whether the penalty algorithm inherits the geometric ergodicity of the MH algorithm (with a slower rate) is worth investigating. The following lemma, which uses an intermediate result of Proposition \ref{prop: ergodicity of the penalty algorithm}, will be useful to establish geometric ergodicity of the penalty algorithm.
\begin{lem} \label{lem: bounded rejection probability for the penalty algorithm}
Suppose the MH algorithm targeting $\pi$ on $\Theta$ has the rejection probability function $\rho(\theta)$. Also, suppose the penalty algorithm for $\pi$ that uses the same proposal kernel and variance $\sigma^{2}: \Theta^{2} \rightarrow [0, \infty)$ has the rejection probability function $\rho_{\sigma}(x)$. If $\esssup \rho < 1$, the essential supremum being taken with respect to $\pi$, and $\sup_{\theta, \theta'} \sigma^{2}(\theta, \theta') = B < \infty$, then $\esssup \rho_{\sigma} < 1$.
\end{lem}
\begin{proof}
We have shown in Proposition \ref{prop: ergodicity of the penalty algorithm} that if $\sup_{\theta, \theta'} \sigma^{2}(\theta, \theta') = B < \infty$ then $\alpha_{\sigma}(\theta, \theta') \geq \kappa \alpha(\theta, \theta')$ with $\kappa = 1 - \Phi(B/2)$.
\begin{align*}
\rho_{\sigma}(\theta) &= 1 - \int \alpha_{\sigma}(\theta, \theta') q(d\theta'|\theta) \\
& \leq 1 - \kappa \int \alpha(\theta, \theta') q(d\theta'|\theta) \\
& = 1 - \kappa (1 - \rho(\theta)) \\
& = (1 - \kappa) + \kappa \rho(\theta)
\end{align*}
Therefore, $\esssup \rho_{\sigma} \leq (1 - \kappa) + \kappa \esssup \rho < 1$.
\end{proof}
Define the sub-stochastic operator $T_{\sigma, a}$ the acceptance part of the kernel for the penalty algorithm, i.e.\ for any $\varphi \in L^{2}(\pi)$
\[
T_{\sigma, a} \varphi(\theta) = \int q(d\theta' | \theta) \alpha_{\sigma}(\theta, \theta') \varphi(\theta')
\]
We give a result on the geometric ergodicity of the penalty method that holds under the assumption that $T_{\sigma, a}$ is compact. (See \citet{Atchade_and_Perron_2007} for a discussion on the compactness assumption.)
\begin{prop} \label{prop: geometric ergodicity}
Assume $T_{\sigma, a}$ is compact. If the MH algorithm targeting $\pi$ is geometrically ergodic, then the penalty algorithm using a bounded $\sigma^{2}(\theta, \theta')$ is also geometrically ergodic.
\end{prop}
\begin{proof}
The geometric ergodicity of the MH algorithm implies $\esssup \rho < 1$ by \citet[Proposition 5.1.]{Roberts_and_Tweedie_1996}. The boundedness of $\sigma^{2}(\theta, \theta')$ gives $\esssup \rho_{\sigma} < 1$ by Lemma \ref{lem: bounded rejection probability for the penalty algorithm}. Moreover, we observe that $T_{\sigma}$ is self-adjoint since the penalty algorithm is reversible. These facts can be used to establish the geometric ergodicity by verifying the steps of Theorem 2.1.\ of \citet{Atchade_and_Perron_2007} for the kernel of the penalty algorithm.
\end{proof}

\section{Differential privacy of the penalty algorithm} \label{sec: Differential privacy of the penalty algorithm}
We have mentioned above that the MH algorithm is superior to the penalty algorithm that mimics it with an extra noise term. Then, what is the point of using the penalty algorithm in the first place? There are at least two reasons, the second of which is relevant to data privacy.
\begin{enumerate}
\item $\lambda(\theta, \theta')$ may be intractable or too expensive to compute and instead an \emph{approximately} normal estimate of it may be used (in which case the resulting algorithm becomes an \emph{approximate} penalty algorithm in that the target density is no more $\pi$). This is in fact how the algorithm has been used in the literature so far.
\item Adding noise to $\lambda(\theta, \theta')$ may help with preserving some sort of \emph{data privacy} in a Bayesian framework where $\pi$, hence $\lambda(\theta, \theta')$, depends on the data to be conditioned on.
\end{enumerate}
The second argument is the main motivation of this paper and needs elaboration with a model that we will study for the rest of the paper. Suppose that we are interested in Bayesian inference of $\theta$ where $\pi$ is the posterior distribution given the data $y_{1:n}$ of size $n$ with prior distribution having a probability density function $\eta$,
\[
\pi(\theta) := p(\theta | y_{1:n}) \propto \eta(\theta) p(y_{1:n} | \theta),
\]
and the data samples conditional on $\theta$ are independent, which leads to the following simple factorisation of the data likelihood
\[
p(y_{1:n} | \theta) = \prod_{t = 1}^{n} p(y_{t} | \theta).
\]
Let us rewrite the logarithm of the acceptance ratio as $\lambda_{n}(y_{1:n}, \theta, \theta') = d_{n}(y_{1:n}, \theta, \theta')  + \log \frac{\eta(\theta') q(\theta | \theta')}{\eta(\theta) q(\theta' | \theta)}$ where
\[
d_{n}(y_{1:n}, \theta, \theta') := \sum_{t = 1}^{n} \left[ \log p(y_{t} | \theta')  - \log p(y_{t} | \theta) \right]
\]
where we include $n$ as well as $y_{1:n}$ in the notation to indicate dependency on data. Assuming calculation of $\log \frac{\eta(\theta') q(\theta | \theta')}{\eta(\theta) q(\theta' | \theta)}$ is straightforward, in order to implement the penalty algorithm, one only needs a normally distributed noisy version of $d_{n}(y_{1:n}, \theta, \theta')$ provided that the variance of the noise is known. Therefore, the data owner can feed the analyser with the noisy log-acceptance ratio and ensure a certain degree of data privacy as well as allowing analysis of the data.  Another relevant scenario is where the data are \emph{shared} among a certain number of users whose common interest is inferring $\theta$ and common concern is privacy of their data against each other. As we will discuss later, the penalty algorithm suggests that those users can execute a common algorithm by submitting their contributions to the log acceptance ratio in gaussian noise with known variance.

The important questions to be asked in the scenarios considered above are:
\begin{itemize}
\item How should we choose the noise variance $\sigma_{n}^{2}(\theta, \theta')$ to ensure a certain degree of data privacy?
\item How does the differential privacy of such schemes scale with data size $n$?
\end{itemize}
These questions will be addressed in the rest of this section, mainly by Theorem \ref{thm: Differential privacy of penalty algorithm}. In the following, we will investigate the degree to which differential privacy is ensured by enabling implementation of the penalty algorithm with reasonable convergence properties. We will first give our main result on the differential privacy of the penalty algorithm. We then provide an interpretation of the result and comment about the assumptions. Then we will provide a proof of the main result which relies on the concepts of \emph{advanced composition} \citep{Dwork_et_al_2010} and \emph{Gaussian mechanism} \citep{Dwork_and_Roth_2013}.

\subsection{Reviewing the penalty algorithm in the privacy context} \label{sec: Reviewing the penalty algorithm in the privacy context}
In order to study the differential privacy of the penalty algorithm, it is useful to view its iterations as a sequence of \emph{database access mechanisms} where these mechanisms are  (randomly) called upon outside by an \emph{adversary}. In particular, we need to introduce the concept of the \emph{Gaussian mechanism} \citep{Dwork_and_Roth_2013}. Define the Gaussian mechanism with $\sigma^{2} > 0$ for a function $f: \mathcal{X} \rightarrow \mathbb{R}^{d}$ with some $d \geq 1$ to be an algorithm $\mathcal{A}_{\sigma}: \mathcal{X} \rightarrow \mathbb{R}^{d}$ which outputs a noisy version of $f(X)$ for input $X$:
\[
\mathcal{A}_{\sigma}(X) = \hat{f}(X) \sim \mathcal{N}(f(X), \sigma^{2} I_{d}).
\]
For the penalty algorithm, for $\theta, \theta' \in \Theta$, $n \geq 1$ and a non-negative real number $\sigma > 0$, we define the database access mechanism $\mathcal{M}_{\theta, \theta', n, \sigma}$ to be the Gaussian mechanism with variance $\sigma^{2}$ for the function $d_{n}(\cdot, \theta, \theta'): \mathcal{Y}^{n} \rightarrow \mathbb{R}$. That is, given the data $y_{1:n}$, $\mathcal{M}_{\theta, \theta', n, \sigma}$ returns
\begin{equation} \label{eq: Gaussian mechanism for d}
\hat{d}_{n, \sigma}(y_{1:n}, \theta, \theta') \sim \mathcal{N}(d_{n}(y_{1:n}, \theta, \theta'), \sigma^{2}).
\end{equation}
Equation \eqref{eq: Gaussian mechanism for d} is nothing but the noisy version of the likelihood ratio in the penalty algorithm. 

Given the data $y_{1:n}$, initial value $\theta_{0} \in \Theta$ and variance $\sigma_{n}^{2}(\theta, \theta') = \sigma^{2} > 0$ for all $\theta, \theta' \in \Theta$, let the sequence of generated samples and proposed samples of the penalty algorithm be $\{ \theta_{n, t}, t \geq 1 \}$ and $\{ \theta'_{n, t}, t \geq 1 \}$ respectively. Given the current sample $\theta_{n, t-1} = \theta$ and the proposed sample $\theta'_{n, t} = \theta'$, the next iteration of the penalty algorithm can be viewed as in Algorithm \ref{alg: Iteration t of the penalty method with a database access view}:
\begin{algorithm} 
\caption{\textbf{Iteration $t$ of the penalty method with a database access view}}
\label{alg: Iteration t of the penalty method with a database access view}
\KwIn{$\theta_{n, t-1} = \theta$, $\theta'_{n, t} = \theta'$}
\KwOut{$\theta_{n, t}$, $\theta'_{n, t+1}$, $\hat{d}_{n, \sigma}(y_{1:n}, \theta, \theta')$}
The database access mechanism $\mathcal{M}_{\theta, \theta', n, \sigma}$ returns $\hat{d}_{n, \sigma}(y_{1:n}, \theta, \theta')$ as in \eqref{eq: Gaussian mechanism for d}. \\
The adversary decides on $\theta_{n, t}$ according to the acceptance probability 
\[
\min \left\{ 1, \frac{\eta(\theta') q_{n}(\theta | \theta')}{\eta(\theta) q_{n}(\theta' | \theta)} \exp\left( \hat{d}_{n, \sigma}(y_{1:n}, \theta, \theta') - \sigma^{2}/2 \right) \right\}.
\]
The adversary then samples $\theta'_{n, t+1}$ from $q_{n}(\cdot | \theta_{n, t})$. 
\end{algorithm}

The \emph{view of the adversary} using the penalty algorithm targeting $\pi_{n}$ for $k$ iterations is
\[
\{ \theta_{n, t}, \theta'_{n, t}, \hat{d}_{n, \sigma}(y_{1:n}, \theta_{n, t-1}, \theta'_{n, t}); 1 \leq t \leq k \}.
\]
That is, we consider the scenario in which not only the generated samples but also the estimates of the log acceptance ratios are observed during the run of the penalty algorithm. For simplicity of the analysis and emphasis on the contribution of the accept-reject procedure of the penalty algorithm, we will from here on assume that the proposal distribution $q_{n}$ does not depend on $y_{1:n}$ (more on this assumption later). Under that assumption, observe that only Step 1 of the above description requires access to data. Therefore, the differential privacy of the penalty algorithm lies in the differential privacy of a sequence the Gaussian mechanisms called during the iterations.


\subsection{The main result} \label{sec: The main result}
Denote the joint probability distribution of sampled and proposed variables by $P_{n, \sigma, \theta_{0}}$ ($y_{1:n}$ is omitted from the notation for simplicity) so that for every $k \geq 1$
\[
P_{n, \sigma, \theta_{0}}(d\theta_{n, 1:k} \times d \theta'_{n, 1:k}) = \prod_{t = 1}^{k} q_{n}(d\theta'_{n, t} | \theta_{n, t-1}) \left[ \alpha_{\sigma}(\theta_{n, t-1}, \theta'_{n, t}) \delta_{\theta'_{n, t}}(d\theta_{n, t}) +  (1 - \alpha_{\sigma}(\theta_{n, t-1}, \theta'_{n, t})) \delta_{\theta_{n, t-1}}(d\theta_{n, t}) \right]
\]
Our main theorem requires the following assumption on the structure of the log-likelihood together with the choice of the proposal density to hold $P_{n, \sigma, \theta_{0}}$-a.s.
\begin{asmp} \label{asmp: bounded difference in log-likelihood}
For all $\theta \in \Theta$, there exists an $\alpha > 0$ such that $\nabla_{2} d_{n}(\cdot, \theta, \theta') \leq c n^{-\alpha}$ $P_{n, \sigma, \theta_{0}}$-a.s., where the $\ell_{2}$ sensitivity $\nabla_{2} f$ of the function $f: \mathcal{X} \rightarrow \mathbb{R}^{d}$ for some $d \geq 1$ is defined as
\[
\nabla_{2} f = \sup_{x, y \in \mathcal{X} : h(x, y) \leq 1} \vert\vert f(x) - f(y) \vert\vert_{2}.
\]
\end{asmp}
Theorem \ref{thm: Differential privacy of penalty algorithm} states the number of iterations allowed to have a $(\epsilon, \delta)$ differential privacy with its $\epsilon$ and $\delta$ parameters having desirable orders of magnitude, which are $\mathcal{O}(1)$ and $o(1/n)$, respectively (see the discussion in \citet[Section 2.3]{Dwork_and_Roth_2013}).
\begin{thm} \label{thm: Differential privacy of penalty algorithm}
Assume \ref{asmp: bounded difference in log-likelihood} with some $c > 0$ and $\alpha > 0$. For every $\beta > 0$, $k_{0} > 0$ and $\theta_{0} \in \Theta$, there exists some $\sigma^{2} > 0$ such that it holds $P_{n, \sigma, \theta_{0}}$-a.s.\ that the $k(n) =  \lfloor k_{0} n^{2\alpha}/\log(n) \rfloor$ iterations of the penalty algorithm targeting $\pi_{n}$ and using $\sigma_{n}^{2}(\theta, \theta') = \sigma^{2}$ is $(\epsilon_{n, k(n)}, \delta_{n, k(n)})$-differentially private where 
\begin{align}
\epsilon_{n, k(n)} & \leq 2 \sqrt{k_{0} (2\alpha + \beta) \beta}  \frac{c}{\sigma} + 4 k_{0} \frac{c^{2}}{\sigma^{2}} (2\alpha + \beta) \label{eq: thm epsilon}\\
\delta_{n, k(n)} &\leq 1.25 k_{0} n^{-\beta}/\log(n) + n^{-\beta}. \label{eq: thm delta}
\end{align}
In particular, for $\beta > 1$, we have $\epsilon_{n, k(n)} = \mathcal{O}(1)$ and $\delta_{n, k(n)} = o(1/n)$.
\end{thm}
\paragraph{Remarks:} Theorem \ref{thm: Differential privacy of penalty algorithm} is useful due to its explicit expressions for $k(n)$, $\epsilon_{n, k(n)}$, and $\delta_{n, k(n)}$. (It is possible to read \eqref{eq: thm epsilon} and \eqref{eq: thm delta} as equalities, since if an algorithm is $(\epsilon, \delta)$ differentially private, it is also $(\epsilon', \delta')$ differentially private for $\epsilon' > \epsilon$ and $\delta' > \delta$.) Specifically, it provides the order of how many iterations the penalty algorithm should be run and what $\sigma^{2}$ should be taken to achieve a given degree of differential privacy. In particular, we are interested in $\epsilon_{n, k(n)} = \mathcal{O}(1)$ and $\delta_{n, k(n)} = o(1/n)$ since those are the desired orders for the algorithm to be meaningfully private, see the discussion in \citet[Section 2.3]{Dwork_and_Roth_2013}. 

What the theorem suggests is that by tuning the free parameters $k_{0}$ and $\beta$ we can obtain a `scalable' differential privacy for a non-decreasing (in fact, increasing for large $n$) number of iterations of the penalty algorithm whose variance does not increase with $n$. Note that while $\alpha$ and $c$ depend on the model at hand, we have the freedom choose $k_{0}$, $\beta$. A realistic value for $\alpha$ is 0.5, as  we will discuss in Section \ref{sec: On the assumption on the proposal density} also.

The depicted scenario in which we observe the noisy log-likelihood ratio $\hat{d}_{n, \sigma}(y_{1:n}, \theta, \theta')$ will be more relevant when we discuss the use of the penalty algorithm in a data sharing scenario. In the absence of such a case, one can consider instead a scenario where the view of the adversary is only the posterior samples and the proposed samples $\theta_{n, t}, \theta'_{n, t}$, $t \geq 1$, or even only the posterior samples $\theta_{n, t}$, $t \geq 1$. In such a scenario, we can still use $(\epsilon_{n, k(n)}, \delta_{n, k(n)})$ to describe the differential privacy of the such a scheme. This is because once $\hat{d}_{n, \sigma}(y_{1:n}, \theta, \theta')$ is calculated, the accept-reject step and the proposal step can be seen as random \emph{post-processing steps} that do not require any access to data. This can only result in an algorithm whose differential privacy after $k(n)$ iterations is less than or equal to that $(\epsilon_{n, k(n)}, \delta_{n, k(n)})$ given in \eqref{eq: thm epsilon} and \eqref{eq: thm delta}.

\paragraph{Proof of the main result:} The following theorem by \citet{Dwork_and_Roth_2013} on the Gaussian mechanism plays a significant role in establishing the proof of Theorem \ref{thm: Differential privacy of penalty algorithm}.
\begin{thm} \label{thm: Gaussian mechanism} \textbf{(Gaussian mechanism)}
Let $\epsilon \in (0, 1)$ be arbitrary. The Gaussian mechanism for function $f$ with $\sigma > \nabla_{2}f \sqrt{2 \log(1.25/\delta)}/\epsilon$ is $(\epsilon, \delta)$-differentially private.
\end{thm}
The function $f$ relevant to the penalty algorithm is the $d_{n}(\cdot; \theta, \theta'): \mathcal{Y}^{n} \rightarrow \mathbb{R}$'s, i.e.\ the difference of log likelihood functions, when considered as functions of data samples. We need to cite one more theorem about differential privacy, regarding what happens when a sequence of database access mechanisms with a certain differential privacy is applied \citep{Dwork_et_al_2010}.

\begin{thm} \label{thm: Advanced composition} \textbf{(Advanced composition)}
For all $\epsilon, \delta, \delta' \geq 0$, the class of $(\epsilon, \delta)$-differentially private mechanisms satisfies $(\epsilon', k \delta + \delta')$-differential privacy under $k$-fold adaptive composition for:
\[
\epsilon' = \sqrt{2k \log (1/\delta')} \epsilon  + k \epsilon (e^{\epsilon} - 1).
\]
\end{thm}
The structure of Algorithm \ref{alg: Iteration t of the penalty method with a database access view} suggests that the penalty algorithm can be seen as a $k$-fold adaptive composition, where a database access mechanism with a certain differential privacy is called. This, combined with Theorems \ref{thm: Gaussian mechanism} and \ref{thm: Advanced composition} facilitates the proof of Theorem \ref{thm: Differential privacy of penalty algorithm}. 
\begin{proof}
\emph{(Theorem \ref{thm: Differential privacy of penalty algorithm})} \\
Let $\beta' = 2 \alpha + \beta$ and pick a finite $\sigma > 0$ that satisfies
\begin{equation} \label{eq: lower bound on sigma}
\sigma > \max_{n \geq 1} c n^{-\alpha} \sqrt{2 \beta' \log n}.
\end{equation}
which is possible since the right hand side is finite for each $\alpha > 0$. Next, let $\mathcal{D}_{c, \alpha, n, \sigma} = \{ \mathcal{M}_{\theta, \theta', n, \sigma}: \nabla_{2} d_{n}(\cdot, \theta, \theta') < c n^{-\alpha} \}$ be a family of the Gaussian mechanisms where $\mathcal{M}_{\theta, \theta', n, \sigma}$ was defined earlier. Then, Theorem \ref{thm: Gaussian mechanism} implies that each mechanism in $\mathcal{D}_{c, \alpha, n, \sigma}$ is $(\epsilon, \delta)$-DP where 
\[
\epsilon = \frac{c}{\sigma} n^{-\alpha} \sqrt{2 \beta' \log n}, \quad \delta = 1.25 n^{-\beta'}
\]
since $\sigma > c n^{-\alpha} \sqrt{2 \beta' \log n}$ by \eqref{eq: lower bound on sigma}, hence satisfying the condition of Theorem \ref{thm: Gaussian mechanism}. 

Furthermore, by Theorem \ref{thm: Advanced composition}, we see that for $\delta'_{n} > 0$ and $k \geq 1$, $\mathcal{D}_{c, \alpha, n, \sigma}$ satisfies $(\epsilon_{n, k}, \delta_{n, k})$-DP under $k$-fold adaptive composition
\[
\epsilon_{n, k} \leq \sqrt{4k \beta' \log n \log(1/\delta'_{n})} \frac{c}{\sigma} n^{-\alpha}  + 2 k \frac{c^{2}}{\sigma^{2}} n^{-2\alpha} 2 \beta' \log n, \quad \delta_{n, k} = 1.25 k n^{-\beta'} + \delta'_{n}.
\]
where we used $e^{\epsilon} -1 \leq 2 \epsilon$ since $\epsilon < 1$ (as noted by \citet{Wang_et_al_2015} also). Now, substituting $\delta'_{n} = n^{-\beta}$ and $k = k(n) = \lfloor k_{0} n^{2\alpha}/\log(n) \rfloor $, we have
\begin{align*}
\epsilon_{n, k(n)} &\leq 2 \sqrt{k_{0} \beta' \beta}  \frac{c}{\sigma} + 4 k_{0} \frac{c^{2}}{\sigma^{2}} \beta' \\
\delta_{n, k(n)} &\leq 1.25 k_{0} n^{-\beta}/\log(n) + n^{-\beta}.
\end{align*}
We conclude the proof by substituting $\beta' = 2 \alpha + \beta$ and pointing out that by Assumption \ref{asmp: bounded difference in log-likelihood} each iteration of the penalty algorithm targeting $\pi_{n}$ accesses data via a mechanism from $\mathcal{D}_{c, \alpha, n, \sigma}$ $P_{n, \sigma, \theta_{0}}$-a.s.\
\end{proof}

\subsubsection{On the assumptions about the proposal density} \label{sec: On the assumption on the proposal density}
Although Assumption \ref{asmp: bounded difference in log-likelihood} is not easy to verify, the following two assumptions are more testable and they together imply \ref{asmp: bounded difference in log-likelihood}.
\begin{asmp} \label{asmp: bounded derivative of log-likelihood}
The derivative with respect to $\theta$ of the log-likelihood $\ell(y, \theta)$ exists and for all $y \in \mathcal{Y}$ it is bounded by $M$ in absolute value
\[
\sup_{y \in \mathcal{Y}, \theta \in \Theta} \left\vert \frac{\partial \ell(y, \theta)}{\partial \theta(i)} \right\vert \leq M, \quad i = 1, \ldots, d_{\theta}.
\]
\end{asmp}
\begin{asmp}  \label{asmp: bounded difference in proposal}
There exists an $\alpha > 0$ and $c > 0$ such that  for all $n \geq 1$ we have $|\theta_{n, t}'(i) - \theta_{n, t-1}(i) | < c n^{-\alpha}$ $P_{n, \sigma, \theta_{0}}$-a.s.\ for all $i = 1, \ldots, d_{\theta}$.
\end{asmp}
Assumption \ref{asmp: bounded derivative of log-likelihood} can hold for compact sets for $\theta$ and $y$ and Assumption \ref{asmp: bounded difference in proposal} can be forced, for example, with a truncated random walk. Assumptions \ref{asmp: bounded derivative of log-likelihood} and \ref{asmp: bounded difference in proposal} imply Assumption \ref{asmp: bounded difference in log-likelihood}, since
\[
|\ell(y, \theta') - \ell(y, \theta)| \leq \sum_{i = 1}^{d_{\theta}} |\theta'(i) - \theta(i)| M
\]
by the zeroth order Taylor polynomial approximation and therefore $\nabla_{2} d_{n}(\cdot, \theta, \theta') \leq 2 d_{\theta} n^{-\alpha} M$. For Assumption \ref{asmp: bounded difference in proposal} to be compatible with desirable convergence properties for the algorithm, we would need $\alpha \leq 0.5$, since for any $\alpha > 0.5$ the proposal moves would have too short of a range and this would cause the penalty algorithm to collapse as $n \rightarrow \infty$.

A restrictive assumption we made about $q_{n}$ is its independence from the data $y_{1:n}$. This can be suboptimal when compared to schemes that exploiting the gradient of the posterior $\pi_{n}$ that can only be extracted by using the data. We made the independence assumption merely for emphasising on the differential privacy brought by the Gaussian mechanism that is naturally a part of the penalty algorithm. When $q_{n}$ depends on the data, sampling from $q_{n}$ should be seen as a database access mechanism and therefore the analysis of Algorithm \ref{alg: Iteration t of the penalty method with a database access view} should consider the differential privacy of the proposal step as well. Taking this into account, any proposal scheme whose differential privacy at one iteration is as good as the Gaussian mechanism is welcome, since such a scheme will not distort the order of $\epsilon_{n, k(n)}$ and $\delta_{n, k(n)}$ but will only double these quantities in the worst case. 

\subsubsection{Differential privacy of the penalty algorithm in a data sharing context} \label{sec: Differential privacy of the penalty algorithm in a data sharing context}
Consider a scenario where the available data $y_{1:n}$ is shared among $N$ \emph{data owners (sharers)} who are commonly interested in inferring the common parameter $\theta$ using all the data available. Let the portion of the data shared by sharer $i$ be $y^{(i)}$ of size $n_{i}$, with $n_{1} + \ldots + n_{N} = n$, so that $y_{1:n} = (y^{(1)}, \ldots, y^{(N)})$. Suppose these sharers want to implement a penalty algorithm together that targets $\pi_{n}$ as well as guarantees a certain differential privacy for each sharer against the other sharers. In other words, for each sharer all the other sharers are adversaries. This is possible if each sharer contributes to the log-acceptance with added noise. Specifically, given $\theta, \theta' \in \Theta$ as the current and proposed values in an iteration of the penalty algorithm, the contribution of user $i$ shall be
 \begin{equation} \label{eq: data sharer contribution}
\hat{d}_{n, \sigma}^{(i)}(y^{(i)}, \theta, \theta') \sim \mathcal{N}(d_{n_{i}}(y^{(i)} , \theta, \theta'), \sigma^{2}). 
\end{equation}
The steps of this algorithm can be written as in Algorithm \ref{alg: Iteration t of the penalty method within a data sharing context}.
\begin{algorithm}
\caption{\textbf{Iteration $t$ of the penalty method within a data sharing context}}
\label{alg: Iteration t of the penalty method within a data sharing context}
For $i = 1, \ldots, N$, user $i$ returns its contribution in noise, $\hat{d}_{n, \sigma}^{(i)}(y^{(i)}, \theta, \theta')$ as in \eqref{eq: data sharer contribution}.\\
$\theta_{n, t}$ is decided on according to the acceptance probability 
\[
 \min \left\{ 1,  \frac{\eta(\theta') q_{n}(\theta | \theta')}{\eta(\theta) q_{n}(\theta' | \theta)} \exp \left(\sum_{i = 1}^{N} \hat{d}_{n, \sigma}^{(i)}(y^{(i)}, \theta, \theta') - N \sigma^{2}/2 \right) \right\}.
\] 
Then the proposal $\theta'_{n, t+1}$ is sampled from $q_{n}(\cdot | \theta_{n, t})$ for the next iteration.
\end{algorithm}
The view of the adversaries (which are data sharers at the same time) using the penalty algorithm targeting $\pi_{n}$ for $k$ iterations is
\[
\{ \theta_{n, t}, \theta'_{n, t}, \hat{d}_{n, \sigma}^{(i)}(y^{(i)}, \theta_{n, t-1}, \theta'_{n, t}), 1 \leq i \leq N, 1 \leq t \leq k \}.
\]
Note that the $\ell_{2}$ sensitivity of the function that each data sharer reveals in noise is the same as in Section \ref{sec: The main result}. Therefore, with the same $\sigma^{2}$, the result for differential privacy in Theorem \ref{thm: Differential privacy of penalty algorithm} applies for the differential privacy of each user. However, there is a cost to pay: The total variance variance in the log-acceptance ratio is now $N \sigma^{2}$ instead of $\sigma^{2}$ which slows down the mixing of the chain of the penalty algorithm and makes it worse in Peskun sense. However, if $N$ does not increase with $n$, Algorithm \ref{alg: Iteration t of the penalty method within a data sharing context} is viable since then the variance of the total noise in the algorithm will increase not with $n$.

\subsection{Exponential families} \label{sec: Exponential families}
A broad and useful family of distributions for the likelihood are exponential family distributions where the likelihood of the data given $\theta$ is written as 
\[
p(y | \theta) = h(y) g(\theta) \exp[\varphi(\theta)^{T} S(y)]
\]
where $S(y)$ is a $d_{\varphi} \times 1$ vector of \emph{sufficient statistics} of $y$, $\varphi(\theta)$ is the $d_{\varphi} \times 1$ vector of natural parameters of the model, and $h(y)$ is the normalising constant. For the i.i.d.\ model with $n$ data samples, the logarithm of the ratio of log-likelihoods at $\theta$ and $\theta'$ becomes
\begin{equation} \label{eq: d for exponential families}
d_{n}(y_{1:n}, \theta, \theta') = n [\log g(\theta') - \log g(\theta)] + [\varphi(\theta') - \varphi(\theta)]^{T} S_{n}(y_{1:n})
\end{equation}
where $S_{n}(y_{1:n}) = \sum_{i = 1}^{n} S(y_{i})$.

Obviously, the previously discussed settings and stated results also hold for exponential family distributions; however, exponential families enable an \emph{alternative and simpler} way of sharing data in privacy: As one can see from \eqref{eq: d for exponential families}, a noisy version of $d_{n}(y_{1:n}, \theta, \theta')$ can be provided by providing a \emph{noisy version of the sufficient statistic} $S_{n}(y_{1:n})$. In order to implement the penalty algorithm for $k$ iterations, $k$ independent noisy versions of $S_{n}(y_{1:n})$ will be revealed by the data owner, i.e.\
\[
\hat{S}_{n, \xi}^{(t)}(y_{1:n}) \overset{\text{i.i.d.}}{\sim} \mathcal{N}(S_{n, \sigma}(y_{1:n}), \xi_{n} I_{d}), \quad t = 1, \ldots, k.
\]
The resulting log-ratio of likelihood is then
\[
\hat{d}_{n, \xi}(y_{1:n}, \theta, \theta') = n [\log g(\theta') - \log g(\theta)] + [\varphi(\theta') - \varphi(\theta)]^{T} \hat{S}_{n}(y_{1:n})
\]
It can be checked that variance of $\hat{d}_{n, \xi}(y_{1:n}, \theta, \theta')$ can be written as
\begin{equation} \label{eq: penalty variance for exponential families}
\sigma_{n}^{2}(\theta, \theta') = [\varphi(\theta') - \varphi(\theta)]^{T} [\varphi(\theta') - \varphi(\theta)] \xi_{n}.
\end{equation}
The resulting penalty algorithm for exponential families is given in Algorithm \ref{alg: Penalty method for exponential families}. Note that unlike the method in \citet{Foulds_et_al_2016} which adds noise to $S_{n}(y_{1:n})$ only once and then processes it, the penalty algorithm requires independent noisy sufficient statistics as many as the number of iterations it runs for. Another difference from the work of \citet{Foulds_et_al_2016} is that their method is only asymptotically unbiased in terms of the target distribution whereas the penalty method is always unbiased.
\begin{algorithm} 
\caption{\textbf{Iteration $t$ of the penalty method for exponential families}}
\label{alg: Penalty method for exponential families}
\KwIn{$\theta_{n, t-1} = \theta$, $\theta'_{n, t} = \theta'$}
\KwOut{$\theta_{n, t}$, $\theta'_{n, t+1}$, $\hat{S}_{n, \xi}(y_{1:n})$}
The database access mechanism returns and independent $\hat{S}_{n, \sigma}(y_{1:n}) \sim \mathcal{N}(S_{n, \xi}(y_{1:n}), \xi_{n} I_{d})$ as in \eqref{eq: Gaussian mechanism for d}. \\
The adversary calculates $\sigma_{n}^{2}(\theta, \theta') = [\varphi(\theta') - \varphi(\theta)]^{T} [\varphi(\theta') - \varphi(\theta)] \xi_{n}$ and decides on $\theta_{n, t}$ according to the acceptance probability 
\[
\min \left\{ 1,  \frac{\eta(\theta') q_{n}(\theta | \theta')}{\eta(\theta) q_{n}(\theta' | \theta)} \exp \left(\hat{d}_{n, \xi}(y_{1:n}, \theta, \theta') - \sigma_{n}^{2}(\theta, \theta')/2 \right) \right\}.
\]
The adversary then samples $\theta'_{n, t+1}$ from $q_{n}(\cdot | \theta_{n, t})$. 
\end{algorithm}

There are two computation-wise important observations for Algorithm \ref{alg: Penalty method for exponential families}. Firstly, the database access mechanism is much simpler than the general case since it returns independent noisy versions of the same quantity. Secondly, unlike for general models, the computation load does not grow with $n$ for exponential families.

Differential privacy of $k$ iterations of Algorithm \ref{alg: Penalty method for exponential families} can be computed using Theorems \ref{thm: Gaussian mechanism} and $\ref{thm: Advanced composition}$ once we know the $l_{2}$ sensitivity of $S$:
\[
\nabla_{2} S = \sup_{x, y \in \mathcal{Y}} || S(x) - S(y) ||_{2}
\]
On one hand, due to convergence issues discussed in Section \ref{sec: Statistical properties}, it is important to keep $\sigma_{n}^{2}(\theta, \theta)$ non-increasing with $n$. On the other hand, $\sigma_{n}^{2}(\theta, \theta')$ should be large enough to satisfy differential privacy for sufficiently many iterations. Fortunately, Corollary \ref{cor: thm for exponential families} shows that we can do these two things at the same time under Assumption \ref{asmp: bounded difference in proposal - exponential family} which is essentially the counterpart of Assumption \ref{asmp: bounded difference in log-likelihood}.
\begin{asmp}  \label{asmp: bounded difference in proposal - exponential family}
There exists an $\alpha > 0$ and $c > 0$ such that  for all $n \geq 1$ we have $|\varphi(\theta_{n, t}') - \varphi(\theta_{n, t-1}) | < c n^{-\alpha}$ $P_{n, \sigma, \theta_{0}}$-a.s.
\end{asmp}
Again, counterparts of Assumptions \ref{asmp: bounded derivative of log-likelihood} and \ref{asmp: bounded difference in proposal} on $\varphi$ and $q_{n}$ can be made to verify Assumption \ref{asmp: bounded difference in proposal - exponential family}, we do not give them here to avoid repetition.

\begin{cor} \label{cor: thm for exponential families}
Assume \ref{asmp: bounded difference in proposal - exponential family} with some $c > 0$ and $\alpha > 0$. For every $\beta > 0$, $k_{0} > 0$ and $\theta_{0} \in \Theta$, there exists some $\sigma^{2} > 0$ such that it holds $P_{n, \sigma, \theta_{0}}$-a.s.\ that the $k(n) = k_{0} n^{2\alpha}/\log(n)$ iterations of Algorithm \ref{alg: Penalty method for exponential families} targeting $\pi_{n}$ and using $\xi_{n} = \sigma^{2} (\nabla_{2} S)^{2} n^{2\alpha}/ c^{2}$ is $(\epsilon_{n, k(n)}, \delta_{n, k(n)})$-differentially private where $\epsilon_{n, k(n)}$ and $\delta_{n, k(n)}$ are given in \eqref{eq: thm epsilon} and \eqref{eq: thm delta}, respectively. Moreover, the variance of the noise used in the algorithm is bounded: for all $n \geq 1$, $\theta,\theta' \in \Theta$ $\sigma_{n}^{2}(\theta, \theta') \leq (\nabla_{2} S)^{2} \sigma^{2}$.
\end{cor}
\begin{proof}
The proof is similar to that of Theorem \ref{thm: Differential privacy of penalty algorithm}. Let $\beta' = 2 \alpha + \beta$ and pick a finite $\sigma > 0$ such that $\max_{n \geq 1} c n^{-\alpha} \sqrt{2 \beta' \log n}$. The Gaussian mechanism that returns $\hat{S}_{n, \xi}$ with $\xi_{n} = \sigma^{2} (\nabla_{2} S)^{2} n^{2\alpha}/ c^{2}$ is $(\epsilon, \delta)$-DP where $\epsilon = \frac{c}{\sigma} n^{-\alpha} \sqrt{2 \beta' \log n}$ and $\delta = 1.25 n^{-\beta'}$ by Theorem \ref{thm: Gaussian mechanism}. The rest of the proof for differential privacy is the same as in the proof of Theorem \ref{thm: Differential privacy of penalty algorithm}. The bound for $\sigma_{n}^{2}(\theta, \theta')$ should be obvious from the choice made for $\xi_{n}$, equation \eqref{eq: penalty variance for exponential families} and Assumption \ref{asmp: bounded difference in proposal - exponential family}.
\end{proof}

\section{Conclusion} \label{sec: Conclusion}
The main contribution of this paper arises from the simple observation that the penalty algorithm has a built-in noise in its calculations which is not desirable in any other context but can be exploited for data privacy. We have shown that, without letting the variance of the noise increase and hence convergence properties of the algorithm degrade with increasing data size, the differential privacy of the algorithm can scale well with data size. The penalty algorithm has the advantage over the other Bayesian methods for differential privacy in the sense that it is always unbiased in terms of its target distribution.

Our analysis uses the fact that at each iteration a database access mechanism with the same differential privacy is called upon. This is necessary for employing Theorem \ref{thm: Advanced composition} on advanced composition. It would be worth studying what happens when mechanisms with different (but characterisable) differential privacies are called sequentially. This would help us relax our assumption on  the bounded sensitivity of the difference of the log-likelihood.

We have already stated that our algorithm preserves the target density at all iterations but computationally not scalable in data size (unless the posterior distribution has a convenient form such as exponential families).  A computationally scalable version of a privacy preserving penalty algorithm can be developed via subsampling, but this will introduce bias in the target density that may vanish asymptotically in data size under some conditions. This is in fact the topic of an unpublished work carried out by the author (in collaboration).

\section*{Acknowledgement} \label{sec: Acknowledgement}
Thanks to Canan Pehlivan and Taylan Cemgil for their careful reading of this paper.

\bibliographystyle{apalike}
{\footnotesize
\bibliography{myrefs_thesis}}

\appendix
\section{Appendix} \label{sec: Appendix}
\begin{proof} (Proposition \ref{prop: performance w.r.t. penalty term})
Dropping $\theta, \theta'$ from the notation, we will show that the derivative of $\alpha_{\sigma}$ w.r.t.\ $\sigma$ is strictly negative.
\begin{align*}
\alpha_{\sigma} & = \int_{-\infty}^{\infty} \min \{ 1, \exp( r + \sigma x - \sigma^{2}/2 ) \} \phi(x)dx \\
& = \int_{\sigma/2 - r/\sigma}^{\infty} \phi(x)dx +  \int_{-\infty}^{\sigma/2 - r/\sigma} \exp(r + \sigma x - \sigma^{2}/2) \phi(x)dx \\
& = \int_{\sigma/2 - r/\sigma}^{\infty} \phi(x)dx + \exp(r) \int_{-\infty}^{-\sigma/2 - r/\sigma} \phi(x)dx
\end{align*}
where $\phi$ is the probability density of the standard normal distribution. Now,
\begin{align*}
\frac{d\alpha_{\sigma}}{d \sigma} &= -\phi(\sigma/2 - r/\sigma) (0.5 + r/\sigma^{2}) + \exp(r) \phi(-\sigma/2 - r/\sigma) (-0.5 + r/\sigma^{2}) \\
& = -\phi(\sigma/2 - r/\sigma) (0.5 + r/\sigma^{2}) + \phi(\sigma/2 - r/\sigma) (-0.5 + r/\sigma^{2}) \\
& = -\phi(\sigma/2 - r/\sigma)
\end{align*}
which is negative.
\end{proof}

\end{document}